\documentclass[11pt,reqno]{amsart}

\usepackage{amsmath,amssymb}
\usepackage{graphicx}
\usepackage{xcolor}
\usepackage[colorlinks=true,linkcolor=blue!60!black,citecolor=blue!60!black,urlcolor=blue!60!black]{hyperref}
\usepackage{bm}
\usepackage[margin=1.1in]{geometry}
\usepackage[square,numbers,sort]{natbib}
\newtheorem{theorem}{Theorem}[section]
\newtheorem{proposition}[theorem]{Proposition}
\newtheorem{corollary}[theorem]{Corollary}
\newtheorem{lemma}[theorem]{Lemma}

\newtheorem{mainthm}{Theorem}

\newtheorem*{informal}{Main Theorem (informal)}

\theoremstyle{remark}
\newtheorem{remark}[theorem]{Remark}

\newcommand{\R}{\mathbb{R}}
\newcommand{\T}{\mathbb{T}}
\newcommand{\E}{\mathbb{E}}
\newcommand{\Z}{\mathbb{Z}}
\newcommand{\PP}{\mathbb{P}}
\newcommand{\cJ}{\mathcal{J}}
\newcommand{\cI}{\mathcal{I}}
\newcommand{\cK}{\mathcal{K}}

\title[High dimensional basins of attraction]{The tentacles landscape: geometric properties of high-dimensional basins of attraction}

\author{Pablo Groisman}
\address{Departamento de Matem\'atica, Facultad de Ciencias Exactas y
Naturales, Universidad de Buenos Aires, and IMAS-UBA-CONICET, Buenos
Aires, Argentina}
\email{pgroisma@dm.uba.ar}

\subjclass[2020]{Primary 34D06, 34C15; Secondary 60F05, 37C75}
\keywords{Basins of attraction, synchronization, Kuramoto model,
winding number, multistability, high-dimensional geometry}

\begin{document}

\begin{abstract}
Basins of attraction in multistable high dimensional dynamical systems are expected to have universal features but very little has been proved rigorously. We consider phase oscillators coupled according to a cycle graph, $\dot\theta_i = f(\theta_{i+1}-\theta_i) + f(\theta_{i-1}-\theta_i)$, with coupling $f$ that is $C^1$, odd, $2\pi$-periodic, and strictly increasing on $(-\pi,\pi)$. We prove the full ``octopus'' picture of the basins of attraction observed numerically by Zhang and Strogatz [Phys.\ Rev.\ Lett.\ 127 (2021) 194101] and, beyond this model, across a wide family of high-dimensional multistable systems.

In our case, we have a family of stable equilibria that can be indexed by their winding number $q \in \Z\cap(-n/2,n/2)$. Basin volumes obey a Gaussian law $\mu(\cK_q)=\sqrt{6/(\pi n)}\,e^{-6q^2/n}(1+o(1))$ in the winding
number. The distance from a uniform sample to its attractor, when divided by $\sqrt n$, concentrates at $\sqrt{\pi^2/3}\approx 1.814$. Along almost every straight line through any twisted state, the ray enters every other basin infinitely many times, with frequencies given by the basin volumes. The
inscribed ball at a twisted state has radius $({\pi}/{\sqrt2})(1-{2|q|}/{n})$ for every $n$, while as $n\to\infty$, a typical ray travels distance $(\pi/2)\sqrt{n/\log n}$ before first leaving the basin: the head of
the octopus is sharply anisotropic.
\end{abstract}

\maketitle

%% ============================================================
\section{Introduction}\label{sec:intro}
%% ============================================================

Multistable dynamical systems pervade science and engineering: power
grids settle into one of many synchronous
configurations~\cite{menck2014dead,dorfler2013synchronization}, neural
networks descend to one of exponentially many low-loss
minima~\cite{li2018visualizing,choromanska2015loss}, ecosystems
organize into competing
equilibria~\cite{biroli2018marginally,altieri2021properties}, proteins can assume a large number of nearly isoenergetic conformations \cite{Frauenfelder1991} and
amorphous solids freeze into one of countless jammed
packings~\cite{stillinger2015energy,charbonneau2017glass}. In each
case, which outcome the system reaches is determined not by the
attractors themselves but by the \emph{basins of attraction}---the
regions of state space that flow to each one. Basin volumes set the probabilities of relaxation outcomes---the \emph{basin stability} of each state, in the sense of Menck \emph{et al.}~\cite{menck2013basin}, a nonlocal measure complementing linear stability---and basin shapes control sensitivity to perturbations and the routes between
attractors~\cite{milnor1985concept,ott2002chaos,aguirre2009fractal}.

In low dimensions, basins can be visualized and understood. In high
dimensions, the picture has been less clear. Geometric intuition
fails in unfamiliar ways~\cite{Donoho,Talagrand,AGM}: the volume of a
high-dimensional shape concentrates in unexpected regions, simple
objects like hypercubes have most of their volume in narrow corners
far from the center, and the success of methods that work in low
dimensions can collapse without warning. Yet across a striking range
of systems, numerical evidence has converged on a common geometric
picture. Studies of jammed sphere packings reported basins shaped
like high-dimensional octopi, with most of their volume in long
filamentary tentacles rather than near the
attractor~\cite{ashwin2012calculations,martiniani2016structural,martiniani2017thesis}.
Independent work on the Kuramoto model on a cycle graph found the
same picture for synchronization
basins~\cite{martiniani2017thesis,zhang2021basins}, where Zhang and
Strogatz coined the ``basins with tentacles'' name and reported a
striking set of geometric features: basin sizes following a Gaussian
law in the winding number, a universal distribution of distances
between random states and their attractors, and the complete failure
of local hypercube approximations to capture basin volume. A
perspective by Casiulis and Martiniani~\cite{casiulis2023when}
articulated this as a universality picture spanning glasses, granular
packings, neural networks, and dynamical systems more broadly.

All of this evidence has been numerical, and recent work has shown
that even the numerical evidence is more fragile than was understood.
Suryadevara, Casiulis, and Martiniani~\cite{suryadevara2025basins}
demonstrated that standard optimizers used to identify basins in
soft-sphere packings produce systematically wrong answers in
moderately high dimensions, attributing initial conditions to the
wrong basins and creating spurious fractal-like geometry. Apparent
power-law features in earlier studies turned out to be artifacts of
inadequate sampling of broad distributions; the basins themselves,
when probed correctly, seem to be smooth structures with well-defined
length scales. The high-dimensional curse cuts both ways: it makes
the geometry hard to picture, \emph{and} it makes the numerics that
picture it unreliable. A rigorous instance of the octopus
picture---in any concrete multistable model---has been missing.

In this paper we provide one. Consider $n$ identical phase
oscillators coupled according to a cycle graph. In a co-rotating frame, the system reads,
\begin{equation}
\dot\theta_i = f(\theta_{i+1}-\theta_i) + f(\theta_{i-1}-\theta_i),
\qquad i=1,\dots,n, \quad (\text{mod }n).
\label{eq:model-intro}
\end{equation}
The coupling function $f$ is assumed to be $C^1$, odd, $2\pi$-periodic, and \emph{strictly increasing on
$(-\pi,\pi)$} (hypotheses (H1)--(H3) below).
% The standard Kuramoto coupling $f=\sin$ satisfies the first three conditions but narrowly
% fails the fourth: $\sin$ is increasing only on $(-\pi/2,\pi/2)$. This
% single hypothesis change has a decisive consequence.
Note that, denoting ${\bm \theta} = (\theta_1, \dots, \theta_n)$ our system \eqref{eq:model-intro} reads
\begin{equation}
\dot {\bm\theta} = -\nabla E_n({\bm \theta}).
\label{eq:gradient}
\end{equation}
The energy landscape underlying our model, $E_n(\bm\theta)=\sum_j F(\theta_{j+1}-\theta_j)$ with $F'=f$, belongs
to a broader family of sum-over-a-convex-function-of-edge-differences energies, that are ubiquitous.

% VA DESPUES.
%
% In most cases, places: the self-attention dynamics of transformer networks are a gradient flow of a structurally analogous
% energy~\cite{geshkovski2024emergence}, and the AKOrN architecture of
% Miyato et al.~\cite{miyato2025akorn} uses Kuramoto-type oscillator
% coupling explicitly as the forward pass of a neural network. We
% return to these connections in Section~\ref{sec:discussion}.

It is easy to check that the set of equilibria of \eqref{eq:model-intro} (critical points of $E_n$) in $\cJ$ is given by,
\[
\theta_i^{(q)} = \frac{2\pi q i}{n} + c , \qquad i=1,\dots n,
\]
for $|q|=0,1,2,\dots,\lfloor n/2 \rfloor$. They are called {\em twisted states}.

Writing $\eta_i=\theta_{i+1}-\theta_i$ for the phase differences, we define, for every ${\bm \theta}$ with $|\eta_i|\ne \pi$, $i=1,\dots, n$, the winding number
$$
I({\bm \theta}) =\frac{1}{2\pi}\sum_i\eta_i.
$$
Observe that $I({\bm \theta}^{(q)})=q$. The {\em basin of attraction} of each of them is given by
\[
 \mathcal K_q = \{{\bm \theta_0} \colon {\bm \theta}^{\bm \theta_0}_t \to {\bm \theta}^{(q)}, \text{ as } t\to\infty\}.
\]
Our main result is the following geometric description of the basins
$\cK_q\subset\T^n$ of the $q$-twisted states; precise statements,
with the exact asymptotic regime and mode of convergence for each
part, are given as Theorems~\ref{thm:gaussian}--\ref{thm:head} in
Section~\ref{sec:results}.

\begin{informal}
Let $\mu$ be the uniform probability measure on $\T^n$ and let
$\cK_q$ be the basin of attraction of the $q$-twisted state. Under
hypotheses \textup{(H1)--(H3)}:
\begin{enumerate}
\item[(i)]\textup{(Flow invariance and conservation law.)} For every $n$, the region
$\cJ=\{|\eta_i|<\pi \text{ for all } i\}$ is flow-invariant, and
the winding number is conserved for all $t\ge 0$ along every trajectory starting in $\cJ$ \textup{(Proposition~\ref{prop:invariance})}. As a consequence, ${\bm \theta}^{(q)}$ is stable for every $-n/2 < q < n/2$ \textup{(Corollary~\ref{cor:stability})}.
\item[(ii)] \textup{(Volumes.)} As $n\to\infty$, $\mu(\cK_q)=\sqrt{6/(\pi n)}\,
  e^{-6q^2/n}(1+o(1))$, uniformly for $q=O(\sqrt n)$
  \textup{(Theorem~\ref{thm:gaussian})}.
\item[(iii)] \textup{(Typical distance.)} As $n\to\infty$, the normalized distance from
  a uniform sample of $\cK_q$ to its attractor concentrates at
  $\sqrt{\pi^2/3}\approx 1.814$.
  \textup{(Theorem~\ref{thm:master})}.
\item[(iv)] \textup{(Tentacles.)} As $n\to\infty$, almost every ray from any twisted
  state visits every basin with frequency equal to its volume, entering and exiting each one
  infinitely many times \textup{(Theorem~\ref{thm:tentacles})}.
\item[(v)] \textup{(Head.)} The largest ball inscribed in $\cK_q$ at
  the attractor has radius $({\pi}/{\sqrt2})(1-2|q|n^{-1})$, while a
  typical ray first exits $\cK_q$ at distance approximately $(\pi/2)\sqrt{n/\log n}$, as $n\to\infty$: the head is strongly anisotropic and carries a vanishing fraction of the basin's volume
  \textup{(Theorem~\ref{thm:head})}.
\end{enumerate}
\end{informal}

Combined, (i)--(v) force the basin's mass into the tentacles: the
head is a vanishing fraction, and the bulk concentrates at the
typical distance $\sqrt{\pi^2/3}$, distributed across state space
like any set of comparable measure.

\begin{figure}[t!]
\centering
\includegraphics[width=.95\textwidth]{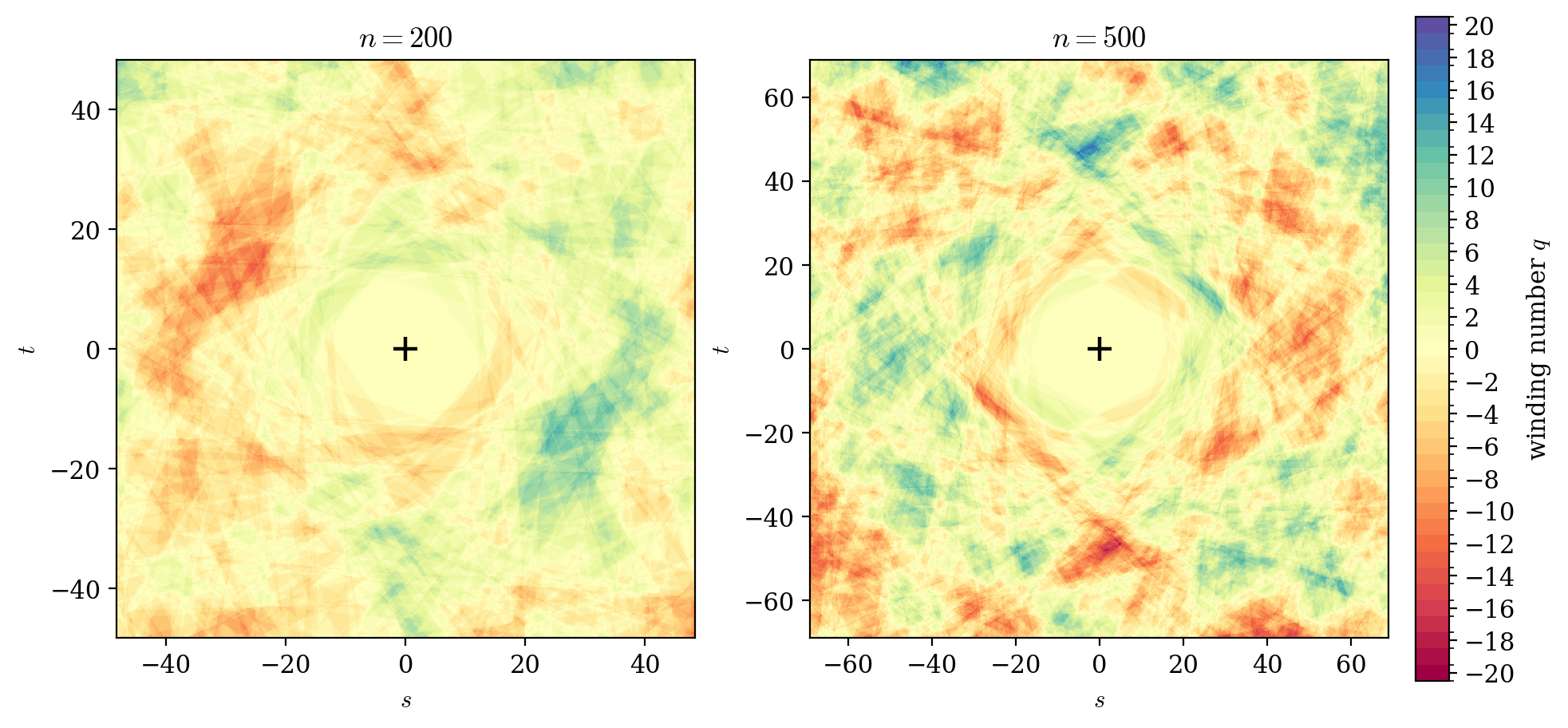}
\caption{Basins on a random 2D affine slice of $\T^n$ through the
$q=0$ twisted state (marked $+$):
$\bm\theta(s,t)=\bm\theta^{(0)}+s\bm v_1+t\bm v_2\pmod{2\pi}$, with
$\bm v_1,\bm v_2$ a random orthonormal pair orthogonal to the gauge
direction $\bm 1$; $n=200$ (left), $500$ (right). Each pixel is colored
by the winding number $I(\bm\theta(s,t))$. On the slice each
constraint $\eta_i=\pm\pi$ is a family of equally spaced parallel
lines, whose arrangement partitions the slice into convex cells of
constant winding number. The trace of each basin $\mathcal K_q$ on the slice is a generically disconnected union of such cells: the same color recurs wherever the
slice re-enters $\mathcal K_q$. Near the attractor the trace of
$\mathcal K_0$ is a single convex cell---the head. A random slice almost
surely avoids the adversarial directions of Theorem~\ref{thm:head}, so the visible cell
extends to distances of order $\sqrt{n/\log n}$ (Theorem~\ref{thm:head}(ii)) rather than the inscribed radius $\pi/\sqrt2$. As $n$ grows the tessellation refines and the fraction of the slice with winding number $0$ shrinks like $n^{-1/2}$ (Theorem \ref{thm:gaussian}). By Proposition~\ref{prop:invariance}
the coupling $f$ does not enter: the winding number, hence the basin
label, is fixed by the initial condition alone. The pictures are inspired by those in \cite{zhang2021basins}.}

\label{fig:slice_atractor}
\end{figure}

% {\bf Earlier analytical
% progress on the standard model~\cite{groisman2025size} had to first
% establish that trajectories enter the smaller invariant region
% $\cI=\{|\eta_i|<\pi/2\}$, a partially rigorous argument that
% controlled basin sizes but said nothing about basin geometry.
% Replacing $\sin$ by a function increasing on the full interval
% removes the analytical obstacle in one step and we need no
% waiting time, no $\log n$ bound, no approximation.}

\subsection{Relation to prior work}\label{sec:prior}

The study of the sync basin on the cycle was initiated by Wiley,
Strogatz, and Girvan~\cite{wiley2006size}, who conjectured the
Gaussian scaling $\mu(\cK_q)\sim e^{-kq^2}$ from numerical evidence for a system like \eqref{eq:model-intro} with $f$ replaced with the $\sin$ function. Delabays, Tyloo, and Jacquod~\cite{delabays2017size} measured basin
sizes through local sampling near each attractor and reported
exponential scaling $e^{-k|q|}$ instead. Zhang and
Strogatz~\cite{zhang2021basins} resolved the discrepancy by
recognizing that the basins are octopus-like: local sampling near the
attractor misses the tentacles entirely, and the tentacles carry
essentially all the volume. The same work reported the typical
distance distribution and the failure of hypercube approximations
that we prove here. Martiniani's PhD Thesis \cite{martiniani2017thesis} proposes a robust numerical protocol,
capable of producing numerical estimates of the total number of stable states and their volumes, for a broad
class of systems that exhibit similar geometry, with particular emphasis on two and three dimensional jammed packings.

On the analytical side, the partially rigorous
argument of~\cite{groisman2025size} established the Gaussian scaling
for the sine coupling via a three-step argument controlling the entry
time into $\cI$; the present paper both removes the
non-rigorous step (in the generalized model) and extends the analysis
from volumes to the full geometry. Static winding-cell
decompositions related to our invariant region appear
in~\cite{delabays2016multistability,jafarpour2022flow}. Our results support the universality picture put
forward in~\cite{zhang2021basins,casiulis2023when}: they constitute
its first rigorous instance.

\subsection{Organization of the paper}

Section~\ref{sec:results} introduces the model and hypotheses and
states the main results. Section~\ref{sec:invariance} proves
well-posedness and the winding-number invariance
(Proposition~\ref{prop:invariance}), and
Section~\ref{sec:stability} the stability of twisted states
(Corollary~\ref{cor:stability}). Sections~\ref{sec:volumes}--\ref{sec:head}
prove Theorems~\ref{thm:gaussian}--\ref{thm:head} in order.
Section~\ref{sec:discussion} discusses universality beyond the cycle
and open problems.

%% ============================================================
\section{The model and main results}\label{sec:results}
%% ============================================================

We consider $n$ identical phase oscillators arranged on a cycle
graph, with each oscillator coupled symmetrically to its two nearest
neighbors. The state of oscillator $i$ is its phase
$\theta_i\in\R/2\pi\Z$, and the system evolves according to \eqref{eq:model-intro}.
% \begin{equation}
% \dot\theta_i = f(\theta_{i+1}-\theta_i) + f(\theta_{i-1}-\theta_i),
% \qquad i=1,\dots,n,
% \label{eq:model}
% \end{equation}
Since~\eqref{eq:model-intro} is invariant under uniform phase shifts
$\bm\theta\mapsto\bm\theta+c\,\bm 1$, we fix this gauge by imposing
$\sum_i\theta_i=0$, a constraint preserved by the dynamics. The
coupling function $f\colon\R\to\R$ is assumed to satisfy:
\begin{itemize}
\item[(H1)] $f\in C^1(-\pi,\pi) \cap C[-\pi,\pi]$;
\item[(H2)] $f$ is odd: $f(-x)=-f(x)$;
%\item[(H3)] $f$ is $2\pi$-periodic;
\item[(H3)] $f'(x)>0$ for all $x\in(-\pi,\pi)$.
\end{itemize}
In the sequel we extend $f$ to the whole line $\R$ periodically.
%The extend function is discontinuous at $2k\pi$, $k\in \Z$.
Conditions (H1)--(H2) are shared with the standard Kuramoto coupling
$f=\sin$. Condition (H3), which $\sin$ fails on $(\pi/2,\pi)$, is
what changes the analysis decisively. Together (H2) and (H3) imply
that, after periodization, $f$ has simple jump discontinuities at the
odd multiples of $\pi$; this is harmless, as trajectories with
typical initial data never reach these points
(Proposition~\ref{prop:invariance}).

It is convenient to work with the phase-difference variables
$\eta_i=\theta_{i+1}-\theta_i$, $i=1,\dots,n$, indices mod $n$. In
these variables the system reads
\begin{equation}
\dot\eta_i = f(\eta_{i+1}) - 2f(\eta_i) + f(\eta_{i-1}),
\label{eq:eta-model}
\end{equation}
with $\sum_{i=1}^n\eta_i\in 2\pi\Z$ enforced by the periodic boundary
conditions. The right-hand side of~\eqref{eq:eta-model} is the
discrete Laplacian of $f(\bm\eta)$ on the cycle.

For trajectories that stay away from the discontinuity set, the
winding number
\begin{equation}
I({\bm \theta}(t)) = \frac{1}{2\pi}\sum_{i=1}^n \eta_i(t)
\label{eq:winding}
\end{equation}
is a well-defined integer. Since each $|\eta_i|/(2\pi)<1/2$, it can
equivalently be computed from any $n-1$ of the phase differences:
\begin{equation}
I({\bm \theta}(t)) = \Bigl[\frac{1}{2\pi}\sum_{i=1}^{n-1}\eta_i(t)\Bigr],
\label{eq:winding-integer-part}
\end{equation}
where $[x]$ denotes the integer closest to $x$. This second form is
the bridge to probability: drawing $\bm\theta(0)$ uniformly on $\T^n$
is equivalent to drawing $\eta_1,\dots,\eta_{n-1}$ i.i.d.\ uniformly
on $[-\pi,\pi]$, the constraint $\sum_i\eta_i\in2\pi\Z$ then
determining $\eta_n$.

Let
\begin{equation}
\cJ = \Bigl\{\bm\eta\in(-\pi,\pi)^n:\ \sum_i\eta_i\in 2\pi\Z\Bigr\}
\label{eq:J-def}
\end{equation}
denote the phase space restricted to configurations for which $I$ is
well-defined. The following proposition is the tool from which all four main theorems are derived.

\begin{proposition}[Winding-number invariance]\label{prop:invariance}
Under hypotheses \textup{(H1)--(H3)}, the region $\cJ$ is positively
invariant under the flow of~\eqref{eq:eta-model}: if
$\bm\eta(0)\in\cJ$, then $\bm\eta(t)\in\cJ$ for all $t\ge 0$.
Consequently the winding number $I({\bm \theta}(t))$ is constant along every
trajectory starting in $\cJ$. The set
$\{\bm\eta\in(-\pi,\pi]^n:\eta_i=\pi\text{ for some }i\}$ has
Lebesgue measure zero in $\T^n$. Consequently, the winding number is conserved for
all $t\ge 0$ for almost every initial condition in state space.
\end{proposition}

Proposition~\ref{prop:invariance} reduces the dynamics to a question
about initial conditions: the basin $\cK_q$ of the $q$-twisted state
coincides, up to a Lebesgue-null set, with $\{\bm\theta(0)\in\T^n: I(\bm \theta(0))=q\}$ (see
Corollary~\ref{cor:stability} for the convergence statement). If
$\bm\theta(0)\sim\mu$, the uniform measure on $\T^n$, then $I(\bm \theta(0))$ is
the rounded sum~\eqref{eq:winding-integer-part} of i.i.d.\ uniform
random variables. Every result in this paper is a consequence of
this reduction. A first consequence combines flow-invariance with the gradient
structure of~\eqref{eq:model-intro}.

\begin{corollary}[Stability of twisted states]\label{cor:stability}
Under hypotheses \textup{(H1)--(H3)}, for every integer $q$ with
$|q|<n/2$ the $q$-twisted state is the unique attractor
of~\eqref{eq:model-intro} within the set $\cK_q$ of configurations with
winding number $q$: every trajectory starting in $\cK_q$ converges to
$\bm\theta^{(q)}$. In particular, all such twisted states are
stable.
\end{corollary}

Observe carefully that, when $n$ is even, for $q=n/2$ we obtain the unstable equilibrium ${\bm \eta}\equiv\pi$, as the reader can check. %The same occurs with states $\bm \theta$ with phase differences $\eta_i \in \{0,\pi\}$ for all $i$ but not all of them equal to zero.

In the standard Kuramoto model, by contrast, twisted states with
$n/4\le|q|<n/2$ are equilibria but unstable: the linearization has a
positive eigenvalue because $\cos(2\pi q/n)<0$ in this range.
%This ``stability gap''~\cite{wiley2006size} is sometimes presented as a
%generic feature of oscillator networks, but it is not: hypothesis
%(H3) eliminates it entirely. It is a special feature of the sine
%coupling near $\pm\pi/2$, where the derivative changes sign.

We can now state the four main theorems. Throughout, $\mu$ denotes
the uniform probability measure on $\T^n$ and $\cK_q$ the basin of
the $q$-twisted state. For the geometric statements we use on $\T^n$
the $\ell^2$ distance on coordinates normalized by $n$,
\begin{equation}
\bar d(\bm\theta,\bm\theta') = \Bigl(\frac1n\sum_{i=1}^n
d_i^2\Bigr)^{1/2},
\label{eq:distance}
\end{equation}
where $d_i\in[0,\pi]$ is the distance on $S^1$ between the $i-$th
coordinates of $\bm\theta$ and $\bm\theta'$, as
in~\cite{zhang2021basins}.

\begin{mainthm}[Gaussian basin volumes]\label{thm:gaussian}
Under hypotheses \textup{(H1)--(H3)}, as $n\to\infty$,
\begin{equation}
\mu(\cK_q) = \sqrt{\frac{6}{\pi n}}\,
\exp\Bigl(-\frac{6q^2}{n}\Bigr)\,(1+o(1)),
\label{eq:gaussian}
\end{equation}
uniformly for $q=O(\sqrt n)$.
\end{mainthm}

Theorem~\ref{thm:gaussian} identifies the constant in the
Wiley--Strogatz--Girvan scaling explicitly, $k=6/n$, and shows it is
universal across the entire family of admissible couplings: basin
volumes depend on the conservation law and the uniform measure, not
on $f$ itself (Remark~\ref{rem:no-finetuning}).

\begin{mainthm}[Typical distance distribution]\label{thm:master}
Fix any $\bm\theta^*\in\T^n$ and sample $\bm\theta$ uniformly on
$\T^n$. Then as $n\to\infty$,
\begin{equation}
\bar d(\bm\theta,\bm\theta^*)\ \to \
\sqrt{\pi^2/3}\ \approx\ 1.814, \quad \mu-\text{almost surely,}
\label{eq:master}
\end{equation}
with Gaussian fluctuations of order $n^{-1/2}$. Moreover,
conditioning on $\bm\theta\in\cK_q$ does not change the limit, for
any $q$ with $|q|=o(\sqrt n)$.
\end{mainthm}

The implication is immediate and counterintuitive: as $n\to\infty$,
the basin of any $q$-twisted state with $|q|=o(\sqrt n)$ lies, up to
vanishing measure, on a sphere of radius $\approx 1.814$ centered at
the attractor. The ``head'' of the basin carries no measure in the
limit; the bulk concentrates at the typical distance to \emph{any}
fixed point of state space, distributed across $\T^n$ like any set of
comparable measure. A complementary statement
(Proposition~\ref{prop:boundary}) shows that a uniform sample lies
arbitrarily close to the boundaries of \emph{many} basins
simultaneously, which explains why local sampling both misses the
bulk and misidentifies what it catches.

\begin{mainthm}[Equidistribution along random rays]\label{thm:tentacles}
Fix $n\ge 3$, an integer $q$, and the $q$-twisted state
$\bm\theta^{(q)}$. Let $\bm v$ be uniform on $S^{n-1}$. Then for
almost every direction $\bm v$, the ray
$\{\bm\theta^{(q)}+\lambda\bm v:\lambda\ge0\}$ is equidistributed on
$\T^n$ with respect to $\mu$. In particular, for every $q'$ with
$\mu(\cK_{q'})>0$,
\begin{equation}
\lim_{T\to\infty}\frac1T\int_0^T
\mathbf 1_{\cK_{q'}}\bigl(\bm\theta^{(q)}+\lambda\bm v\bigr)\,d\lambda
= \mu(\cK_{q'}),
\label{eq:equidist}
\end{equation}
and the ray enters and exits each such $\cK_{q'}$ infinitely many
times.
\end{mainthm}

\begin{figure}[t!]
\centering
\includegraphics[width=.95\textwidth]{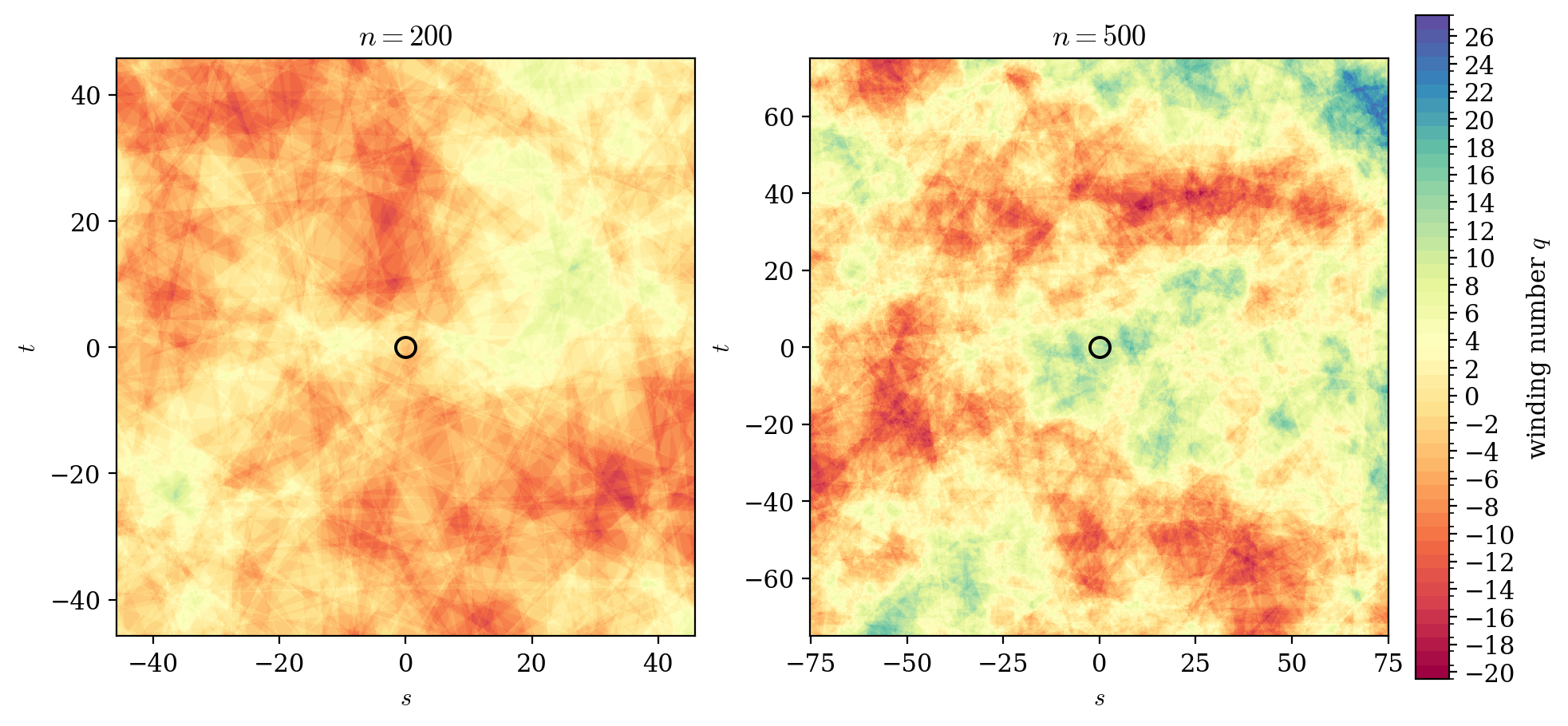}
\caption{The same construction as Figure~\ref{fig:slice_atractor} but through a
\emph{uniformly random} point $\boldsymbol\theta_0$ rather than a twisted
state: $\boldsymbol\theta(s,t)=\boldsymbol\theta_0+s\bm v_1+t\bm v_2
\pmod{2\pi}$, with $\bm v_1,\bm v_2$ a random orthonormal pair;
$n=200$ (left), $500$ (right). Each pixel is colored by the winding number
$I(\boldsymbol\theta(s,t))$. As in Figure~\ref{fig:slice_atractor}, the constraints
$\eta_i=\pm\pi$ cut the slice into convex cells of constant winding number,
and the trace of each basin $\mathcal K_q$ on the slice is a generically disconnected union of such cells, so the same color recurs across the slice. Unlike
Figure~\ref{fig:slice_atractor}, the center is now a generic point: it sits in an
ordinary cell with no distinguished head, and the winding number near it is
its own value $I(\boldsymbol\theta_0)$ rather than $0$. The contrast between
the two figures is the point---the convex head is a feature of the attractors,
not of the ambient space, and a uniform sample lies close to the boundaries
of many basins at once (Proposition~\ref{prop:boundary}). In both images, basins with winding number of order up to $\sqrt n$ are captured, as predicted by Theorem \ref{thm:gaussian}. These pictures also mimic those in \cite{zhang2021basins}.}
\label{fig:slice_random}
\end{figure}

Combining Theorems~\ref{thm:gaussian} and~\ref{thm:tentacles}, the
ray spends a fraction $\sqrt{6/(\pi n)}\,e^{-6q'^2/n}$ of its length
in each basin $\cK_{q'}$; by Corollary~\ref{cor:stability} every
$\cK_{q'}$ with $|q'|<n/2$ has positive measure, so the ray from any
attractor threads through every other basin infinitely often. The
tentacles are not metaphor but quantitative: along a generic line
through any twisted state, every basin is visited with frequency
equal to its volume. This upgrades the ``rays cross many basins''
observation of~\cite{zhang2021basins} to a precise frequency law.

Theorem~\ref{thm:tentacles} does not say at what distance the ray
first leaves the starting basin---the size of the ``head''. Two
natural quantities capture it: the inscribed radius
\begin{equation*}
R_q = \sup\bigl\{r>0:\ B(\bm\theta^{(q)},r)\subset\cK_q\bigr\},
\end{equation*}
where $B$ denotes the Euclidean ball, and, for $\bm v\in S^{n-1}$,
the first-crossing distance
\begin{equation*}
\lambda^*(\bm v) = \inf\bigl\{\lambda>0:\
\bm\theta^{(q)}+\lambda\bm v\notin\cK_q\bigr\},
\end{equation*}
so that $R_q=\inf_{\|\bm v\|=1}\lambda^*(\bm v)$.

\begin{mainthm}[Head size]\label{thm:head}
Let $\bm v$ be uniform on $S^{n-1}$. Then,
\begin{enumerate}
\item[(i)] For $|q|<n/2$, $R_q=\dfrac{\pi-|2\pi q/n|}{\sqrt2}=\dfrac{\pi}{\sqrt2}\Bigl(1-\dfrac{2|q|}{n}\Bigr)$;
\item[(ii)] For $|q|=o(n)$, as $n\to\infty$, $\lambda^*(\bm v)\sqrt{\log n/n}\to\pi/2$, in probability.
\end{enumerate}
\end{mainthm}

The basin therefore has an inscribed ball of radius
$R_q=\tfrac{\pi}{\sqrt2}(1-2|q|/n)$---bounded in $n$ and maximal at $q=0$,
where it equals $\pi/\sqrt2$---but extends a factor $\sqrt{n/\log n}$ farther
along a typical ray: the head is strongly anisotropic, narrow in a few
adversarial directions and thicker in most generic ones (Figure \ref{fig:slice_adversarial}), a non-star-shaped
region whose tentacles reach out to the typical distance $\sqrt{\pi^2/3}\approx1.814$
of Theorem~\ref{thm:master}.

\begin{remark}[Sharpness of the regime in
Theorem~\ref{thm:gaussian}]\label{rem:tail-regime}
The estimate~\eqref{eq:gaussian} is sharp on the regime that carries
essentially all the mass: by Chebyshev's inequality,
$\PP(|I(\bm \theta(0))|>C\sqrt n)\to0$ as $C\to\infty$. Stable twisted states
exist throughout $|q|<n/2$ (Corollary~\ref{cor:stability}), but only
those with $|q|=O(\sqrt n)$ are reached from random initial
conditions in any meaningful sense: even the union of all basins with
$|q|\gg\sqrt n$ has vanishing measure.
\end{remark}

\begin{remark}[No fine-tuning]\label{rem:no-finetuning}
The Gaussian scaling is a property of the cycle topology and of the
uniform initial measure on $\T^n$, not of any particular feature of
$f$ beyond (H1)--(H3). Different choices of $f$ yield the same basin volumes; the
constant $6/n$ in the exponent is universal across the family. The
dependence on $f$ enters only through the within-basin dynamics, not
through the basin volumes or their geometry.
\end{remark}

\begin{figure}[t!]
\centering
\includegraphics[width=.95\textwidth]{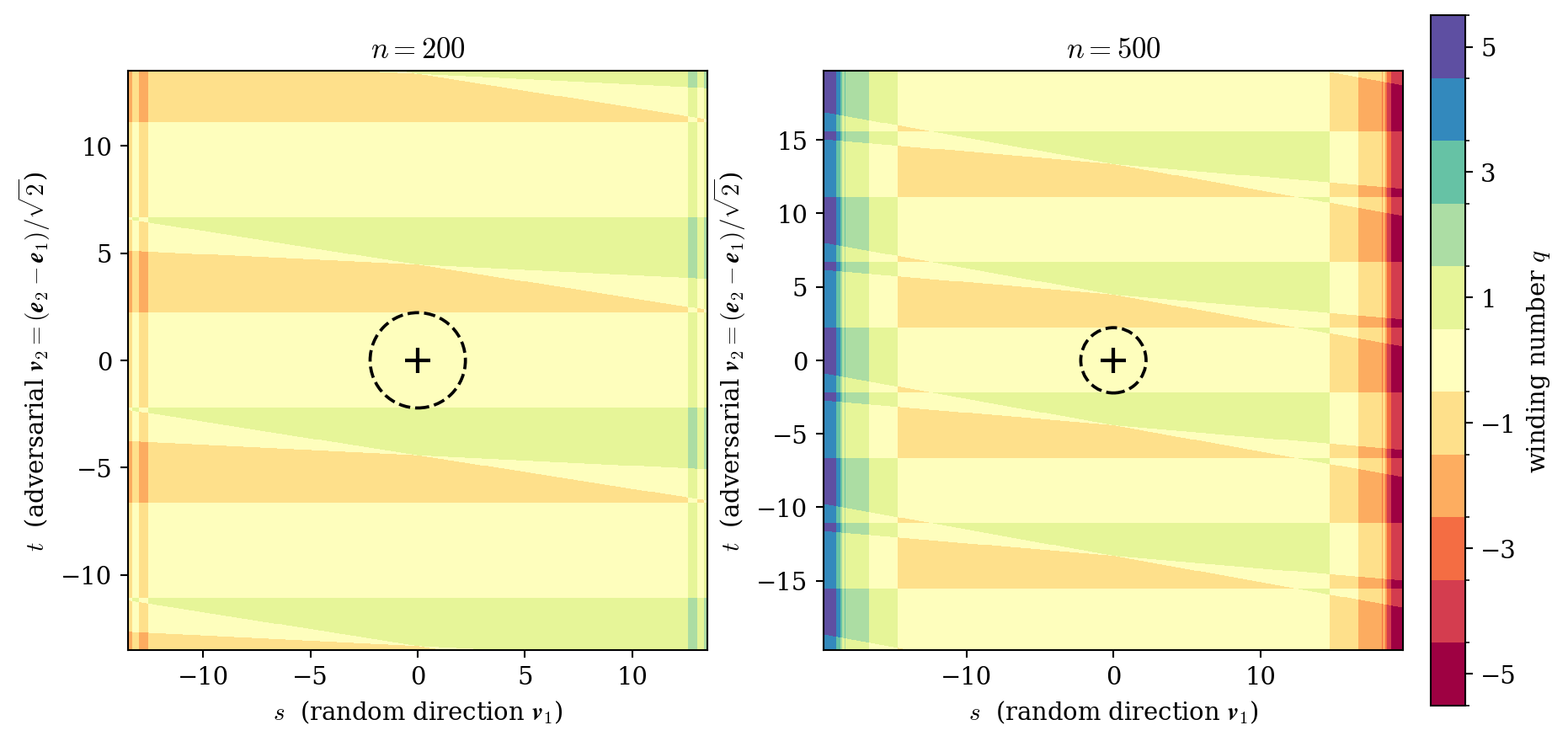}
\caption{Anisotropy of the head, on a slice through the $q=0$ twisted state
(marked $+$) spanned by one random direction and one adversarial direction:
$\boldsymbol\theta(s,t)=s\bm v_1+t\bm v_2\pmod{2\pi}$, with
$\bm v_1$ random (orthogonal to $\bm v_2$ and to the gauge direction
$\bm 1$) and $\bm v_2=(\bm e_2-\bm e_1)/\sqrt2$ the direction
realizing the inscribed radius in Theorem~\ref{thm:head}; $n=200$ (left),
$500$ (right), colored by winding number. The dashed circle is the inscribed
ball of radius $R_0=\pi/\sqrt2$. Moving along $\bm v_2$ drives a single phase difference $\eta_1$ at the
maximal rate $\sqrt2$, so the wall $\eta_1=\pi$ is reached at $t=\pi/\sqrt2$
independent of $n$; moving along the random $\bm v_1$ spreads the
displacement over all $n$ coordinates, so the nearest wall is only reached at
distance $\sim(\pi/2)\sqrt{n/\log n}$. The head is correspondingly elongated;
the surrounding cells are crossed at comparable rates in the two directions.}
\label{fig:slice_adversarial}
\end{figure}

%% ============================================================
\section{Proofs}
In this Section we prove propositions \ref{prop:invariance} and \ref{prop:boundary}, Corollary \ref{cor:stability}, and theorems A–D.

%% ============================================================

\subsection{Well-posedness and winding number invariance}
\label{sec:invariance}

We first record the well-posedness of~\eqref{eq:eta-model} on $\cJ$
and the behavior of the vector field at the boundary.

\begin{lemma}[Well-posedness and boundary extension]\label{lem:wellposed}
Under \textup{(H1)--(H3)}: \textup{(a)} the right-hand side
of~\eqref{eq:eta-model} is $C^1$ on $\cJ$, so for every
$\bm\eta(0)\in\cJ$ there is a unique maximal solution, defined as
long as the trajectory remains in $\cJ$; moreover, the
vector field in~\eqref{eq:eta-model} is continuous in $\overline{\cJ}=[-\pi,\pi]^n\cap\{\sum_i\eta_i\in2\pi\Z\}$; \textup{(b)} if a maximal solution has finite exit time
$t_0<\infty$, then $\bm\eta(t)$ converges as $t\uparrow t_0$ to a
point of $\partial\cJ$, the trajectory is $C^1$ up to $t_0$ with
one-sided derivative at $t_0$.
\end{lemma}

\begin{proof}
(a) is standard ODE theory: (H1) makes $\bm\eta\mapsto
\bigl(f(\eta_{i+1})-2f(\eta_i)+f(\eta_{i-1})\bigr)_i$ a $C^1$ vector
field on the open set $\cJ$ and continuous in $\overline{\cJ}$~\cite{Perko}. For (b), boundedness of the field gives $|\dot{\bm\eta}|\le C$ on $\cJ$, so
$\bm\eta(t)$ is Lipschitz up to $t_0$ and converges to a limit
$\bm\eta(t_0)\in\overline\cJ$; maximality forces
$\bm\eta(t_0)\in\partial\cJ$. Continuity of the field along
the trajectory gives the one-sided derivative at $t_0$.
\end{proof}

\begin{proof}[Proof of Proposition~\ref{prop:invariance}]
We argue by contradiction. Suppose $\bm\eta(0)\in\cJ$ and let
$t_0>0$ be the first time at which the trajectory reaches the
boundary $\partial\cJ$; by Lemma~\ref{lem:wellposed}(b) the limit
$\bm\eta(t_0)\in\partial\cJ$ exists and the trajectory has a
one-sided derivative at $t_0$ given by the field. At time
$t_0$ some component reaches $\pm\pi$; without loss of generality, we assume
$\eta_i(t_0)=\pi$ for some index $i$ (the case $\eta_i(t_0)=-\pi$ is
symmetric by (H2)). Since $\eta_i(t)\in(-\pi,\pi)$ for $t<t_0$ and
$\eta_i(t_0)=\pi$, the one-sided derivative satisfies
\begin{equation}
\dot\eta_i(t_0)\ \ge\ 0.
\label{eq:noexit-ineq}
\end{equation}

We use~\eqref{eq:eta-model} to evaluate the left-hand
side. Both neighboring components $\eta_{i\pm1}(t_0)$ lie in
$[-\pi,\pi]$, so by (H3),
\begin{equation*}
f(\eta_{i+1}(t_0))\le f(\pi), \qquad
f(\eta_{i-1}(t_0))\le f(\pi),
\end{equation*}
with equality if and only if the corresponding
$\eta_{i\pm1}(t_0)=\pi$. Substituting,
\begin{equation*}
\dot\eta_i(t_0) = f(\eta_{i+1}(t_0)) - 2f(\pi) +
f(\eta_{i-1}(t_0)) \ \le\ 0,
\end{equation*}
with equality if and only if $\eta_{i+1}(t_0)=\eta_{i-1}(t_0)=\pi$.
Combining with~\eqref{eq:noexit-ineq} forces $\dot\eta_i(t_0)=0$ and
$\eta_{i\pm1}(t_0)=\pi$. Iterating the same reasoning at the indices
$i\pm1$ and proceeding along the cycle forces $\eta_j(t_0)=\pi$ for
\emph{every} index $j$. But the configuration $\bm\eta\equiv\pi$ is
an equilibrium of the extended field, and uniqueness of solutions
forbids a trajectory starting in $\cJ$ from reaching it in finite
time. This contradiction shows no such $t_0$ exists, and $\cJ$ is
positively invariant.

Invariance of the winding number on $\cJ$ now follows by direct
computation: summing~\eqref{eq:eta-model} over $i$ and using the
cyclic boundary condition,
\begin{equation*}
\frac{d}{dt}\sum_{i=1}^n\eta_i(t) =
\sum_{i=1}^n\bigl[f(\eta_{i+1})-2f(\eta_i)+f(\eta_{i-1})\bigr]=0,
\end{equation*}
since the discrete Laplacian on a cycle telescopes to zero. Hence
$I({\bm \theta}(t))$ is constant in time. Finally, the set of initial conditions
with some $\eta_i=\pi$ is a finite union of hyperplane sections of
$\T^n$, of Lebesgue measure zero, which gives the almost-sure
statement.
\end{proof}

\begin{remark}[Failure for the sine coupling]\label{rem:sin-failure}
The corresponding statement for $f=\sin$ is strictly weaker. There
(H3) holds only on $(-\pi/2,\pi/2)$, and the maximum-principle
argument yields invariance only of the smaller region
$\cI=\{|\eta_i|<\pi/2\}$. Random initial conditions typically fall
outside $\cI$, so for the standard Kuramoto model one must analyze a
transient phase before the winding number is controlled---the
partially rigorous step in~\cite{groisman2025size}, which shows
trajectories enter $\cI$ at time $t\propto\log n$. That suffices for
the basin-volume scaling but not for the geometric statements of
Theorems~\ref{thm:master}--\ref{thm:head}, which concern the basins
as subsets of state space at $t=0$.
\end{remark}

%% ============================================================
\subsection{Stability of twisted states}\label{sec:stability}
%% ============================================================

\begin{proof}[Proof of Corollary~\ref{cor:stability}]
The system~\eqref{eq:model-intro} is the gradient flow of the energy
\begin{equation}
E(\bm\theta) = \sum_{j=1}^n F(\eta_j),
\qquad F'=f,\quad F(0)=0.
\label{eq:energy}
\end{equation}
Hypothesis (H3) makes $F$ strictly convex on $(-\pi,\pi)$. By
Proposition~\ref{prop:invariance}, every trajectory starting in
$\cK_q$ stays in $\cK_q$ and preserves the constraint
$\sum_j\eta_j=2\pi q$. Strict convexity of $F$ together with the
method of Lagrange multipliers gives a unique minimizer of $E$ on
this affine slice: the configuration with all $\eta_j$ equal, that
is, the $q$-twisted state, which requires the common value
$2\pi q/n\in(-\pi,\pi)$, i.e.\ $|q|<n/2$. Since $E$ is bounded below
on $\overline\cJ$ and strictly decreases along non-equilibrium
trajectories, LaSalle's invariance principle~\cite{Perko} implies
every trajectory in $\cK_q$ converges to the $q$-twisted state.
\end{proof}

%% ============================================================
\subsection{Basin volumes.}
\label{sec:volumes}
%% ============================================================

\begin{proof}[Proof of Theorem~\ref{thm:gaussian}]
By Proposition~\ref{prop:invariance} and
Corollary~\ref{cor:stability}, for $\mu$-almost every
$\bm\theta(0)\in\T^n$ the trajectory converges to the twisted state
whose winding number is $I(\bm \theta(0))$. Hence, up to a $\mu$-null set,
\begin{equation*}
\cK_q = \bigl\{\bm\theta\in\T^n:\ I(\bm \theta)=q\bigr\}.
\end{equation*}
Sampling $\bm\theta$ uniformly on $\T^n$ is equivalent to sampling
$\eta_1,\dots,\eta_{n-1}$ i.i.d.\ uniformly on $(-\pi,\pi]$ and
setting $\eta_n=-\sum_{i=1}^{n-1}\eta_i\pmod{2\pi}$. Define
$X_i=\eta_i/(2\pi)$, so the $X_i$ are i.i.d.\ uniform on
$(-\tfrac12,\tfrac12]$ with $\E X_i=0$, $\operatorname{Var}X_i=1/12$,
and, by~\eqref{eq:winding-integer-part},
\begin{equation*}
I(\bm \theta) = \Bigl[\,S_{n-1}\,\Bigr], \qquad
S_{n-1}=\sum_{i=1}^{n-1}X_i .
\end{equation*}
The normalized sum $\sqrt{\frac{12}{n-1}}S_{n-1}$ has a bounded continuous density for $n\ge 3$, so
the local central limit theorem for densities~\cite[Ch.~VII]{Petrov}
gives
\begin{equation*}
\sup_{x\in\R}\ \Bigl|\, p_{n}(x) - \varphi\bigl(x\bigr)\Bigr|
\ \longrightarrow\ 0,
%, \qquad \sigma_n^2=\frac{n-1}{12},
\end{equation*}
where $p_n$ is the density of $\sqrt{\frac{12}{n-1}}S_{n-1}$ and $\varphi$ the standard
Gaussian density. Hence, uniformly for integer $q=O(\sqrt n)$,
\begin{equation*}
\mu(\cK_q) = \PP\bigl([S_{n-1}]=q\bigr)
= \int_{\sqrt{\frac{12}{n-1}}(q-\frac12)}^{\sqrt{\frac{12}{n-1}}(q+\frac12)} p_n(x)\,dx
= \frac{1}{\sqrt{2\pi(n-1)/12}}\,
e^{-6q^2/(n-1)}\bigl(1+o(1)\bigr),
\end{equation*}
which is~\eqref{eq:gaussian}.
\end{proof}

Theorem~\ref{thm:gaussian} confirms the original
Wiley--Strogatz--Girvan conjecture~\cite{wiley2006size} and the
global form of the numerical observations of Zhang and
Strogatz~\cite{zhang2021basins}, while bypassing the
local-versus-global tension that produced the apparent disagreement
between~\cite{wiley2006size} and~\cite{delabays2017size}: the
Gaussian law is a global statement about the entire basin, and is
invisible to local sampling.

%% ============================================================
\subsection{The typical distance distribution.}\label{sec:master}
%% ============================================================

\begin{proof}[Proof of Theorem~\ref{thm:master}]
By rotational invariance of $\mu$ on $\T^n$ we may take
$\bm\theta^*=\bm 0$, so the coordinate distances $d_i$ are i.i.d.\
uniform on $[0,\pi]$ with $\E(d_i^2)=\pi^2/3$ and finite variance.
The strong law of large numbers gives
$\frac1n\sum_id_i^2\xrightarrow{\mathrm{a.s.}}\pi^2/3$, hence
\eqref{eq:master}, and the central limit theorem gives the Gaussian
fluctuations of order $n^{-1/2}$.

For the conditional statement, fix $q$ with $|q|=o(\sqrt n)$ and
condition on the event $\{I(\bm\theta)=q\}$, which by
Theorem~\ref{thm:gaussian} has probability of order $n^{-1/2}$.
%% TODO PG: this conditional step deserves a full argument in the math
%% version (a referee will poke exactly here). Suggested route: a
%% bivariate local CLT for the pair (S_{n-1}, n^{-1} sum d_i^2), or an
%% exchangeable-tilting argument showing the conditional law of
%% (eta_1, ..., eta_{n-1}) given [S_{n-1}] = q is contiguous to the
%% unconditional law on the scale relevant for the LLN. The PRL/PNAS
%% versions cite Petrov Ch. VII and assert the extension is
%% straightforward; here we should write it. I left the statement as
%% in the PRL version; the proof below is the outline to be expanded.
Conditioning on $[S_{n-1}]=q$ fixes one linear functional of
$(\eta_1,\dots,\eta_{n-1})$ within its typical range; a bivariate
local limit theorem for the pair
$\bigl(S_{n-1},\,\sum_{i}d_i^2\bigr)$~\cite[Ch.~VII]{Petrov} shows
that the conditional law of $n^{-1}\sum_id_i^2$ still concentrates at
$\pi^2/3$, with the same first-order fluctuations.
\end{proof}

The boundaries of the basins $\cK_q$ are the configurations at which
the winding number is ill-defined: those with some $\eta_i=\pi$
(with $\pm\pi$ identified on the circle). The next observation
explains why local hypercube approximations cannot capture basin
volume.

\begin{proposition}[Boundary proximity]\label{prop:boundary}
Let $\bm\theta$ be uniform on $\T^n$ and $\delta\in(0,\pi)$. The
number of phase differences $\eta_i$, $1\le i\le n-1$, within
$\delta$ of $\pi$ has distribution $\operatorname{Binomial}(n-1,\delta/\pi)$. In particular, if $\delta=\delta_n$ with
$n\delta_n\to\infty$, then with probability tending to one at least
$n\delta_n/(2\pi)$ of the $\eta_i$ lie within $\delta_n$ of $\pi$, a
number that diverges with $n$.
\end{proposition}

\begin{proof}
Each $\eta_i$ is uniform on $(-\pi,\pi]$, so
$\PP(|\eta_i-\pi|<\delta)=\delta/\pi$ and the count is binomial with
mean $(n-1)\delta/\pi$. When $n\delta_n\to\infty$ the mean diverges
and the claim follows from concentration of the binomial
distribution (Chebyshev suffices).
\end{proof}

A random point of state space therefore lies arbitrarily close to
the boundary of some $\cK_q$, and indeed near many such boundaries
simultaneously. Combined with Theorem~\ref{thm:master}: local
samples from $\cK_q$ near $\bm\theta^{(q)}$ form a vanishing fraction
of the basin's volume \emph{and} sit close to many boundaries with
other basins. Local measurements miss the bulk and misidentify what
they catch.

%%=============================================================
\subsection{Boundary crossings}
%%============================================================

% Theorem \ref{thm:tentacles} states that a random ray visits all the basins with frequencies proportional to their volume, however this phenomena is not a characteristic property of the high dimensionality of the space. In fact, it also happens for any partition of the torus in any dimension. Think of example in the two dimensional torus $[-\pi,\pi)^2$ and the set $\{\theta_1 <0\}$ and its complement. A random ray crosses the boundaries infinitely many times and spens half the time on each set. But this happens because the (infinite) ray wraps infinitely many time around the torus, not because the sets are ``all around''. So, to describe high dimensional basis properly, we need to distinguish from this example. That is the following of the next corollary.

Theorem~\ref{thm:tentacles} states that a random ray visits every basin
with frequency equal to its volume. This phenomenon, however, is not by
itself a signature of high dimension: it holds for \emph{any} partition of
the torus into sets of positive measure, in any dimension. Consider, for
instance, the two-dimensional torus and the partition into the two halves
$\{\theta_1<0\}$ and $\{\theta_1\ge0\}$. A generic ray crosses the
boundary infinitely many times and spends half of its length in each set
--- but only because the infinite ray wraps around the torus infinitely
many times, not because either set is spread ``all around'' the space. To
capture what is genuinely high-dimensional in the geometry of the basins,
one must look at a finite segment, short enough that no coordinate wraps
around its circle. The next proposition does this: already the geodesic
segment from the attractor to a typical point of state space --- along
which every coordinate moves by less than $\pi$ --- crosses $n/4 + o(n)$
basin boundaries. In the two-dimensional example above, the corresponding
count is at most one.

\begin{proposition}[Boundary crossings along a typical segment]\label{cor:crossings}
Fix $\bm\theta^*\in\T^n$ with phase differences
$c_i\in(-\pi,\pi)$, $i=1,\dots,n$, and let $\tilde{\bm\theta}\sim\mu$ be
uniform and independent. Let $[\bm\theta^*,\tilde{\bm\theta}]$ denote the
geodesic segment $\bm\theta(s)=\bm\theta^*+s\bm\delta$, $s\in[0,1]$, where
$\delta_i\in(-\pi,\pi]$ is the signed displacement on $S^1$ from
$\theta^*_i$ to $\tilde\theta_i$, and let
\[
N_n \;=\; \#\bigl\{(i,s):\ 1\le i\le n,\ s\in(0,1),\
\eta_i(s)\in\pi+2\pi\Z\bigr\}
\]
be the number of basin-boundary crossings along the segment. Then, writing $m_n=\dfrac14+\dfrac{1}{4\pi^2 n}\displaystyle\sum_{i=1}^n c_i^2$,
for every $n\ge3$ and every $\varepsilon>0$,
\begin{equation}\label{eq:crossings}
\PP\Bigl(\Bigl|\frac{N_n}{n}-m_n\Bigr|\ge\varepsilon\Bigr)
\;\le\; 2\,e^{-\varepsilon^2 n/2}.
\end{equation}
% in particular $N_n/n-m_n\to0$ in probability, and completely: the
% probabilities in \eqref{eq:crossings} are summable in $n$, so the
% convergence is almost sure under any realization of the systems on a
% common probability space.
% Then
% \begin{equation}\label{eq:crossings}
% \frac{N_n}{n}\;-\;\Bigl(\frac14+\frac{1}{4\pi^2 n}\sum_{i=1}^n c_i^2\Bigr)
% \;\xrightarrow{\ \PP\ }\;0, \qquad n\to\infty .
% \end{equation}
In particular:
\begin{enumerate}
\item[(i)] if $\bm\theta^*=\bm\theta^{(q)}$ with $|q|=o(n)$, then
$N_n/n\to 1/4$ in probability;
\item[(ii)] if $\bm\theta^*$ is itself uniform on $\T^n$ and independent of
$\tilde{\bm\theta}$, then $N_n/n\to 1/3$ in probability.
\end{enumerate}
\end{proposition}

\begin{remark}\label{rem:onecopy}
The terminal winding number satisfies $\E I(\tilde{\bm\theta})=0$ and is
$O(\sqrt n)$ in probability, while the number of
$\pm1$ steps of $I(\bm\theta(s))$ along the segment grows like $n/4$: the
steps cancel massively. Note also that the attractor minimizes the
crossing count: from a twisted state the constant is $1/4$, while between
two independent uniform points it is $1/3$.
\end{remark}

\begin{proof}
Along the segment the phase differences are affine,
$\eta_i(s)=c_i+s\,w_i$ with $w_i=\delta_{i+1}-\delta_i\in(-2\pi,2\pi)$.
Note that it is the \emph{unwrapped} increment $w_i$ that governs
crossings: its reduction mod $2\pi$ is uniform on $(-\pi,\pi]$, but a
crossing is a property of the swept interval from $c_i$ to $c_i+w_i$ in
the lift, and replacing $w_i$ by $w_i\mp2\pi$ changes the number of points
of $\pi+2\pi\Z$ in that interval by one.

Since $|w_i|<2\pi$, the swept interval contains at most one point of
$\pi+2\pi\Z$, so each bond contributes at most one crossing, and it does
so if and only if
\[
\delta_{i+1}-\delta_i>\pi-c_i
\qquad\text{or}\qquad
\delta_{i+1}-\delta_i<-(\pi+c_i).
\]
Because $\tilde{\bm\theta}$ is uniform, the pair $(\delta_i,\delta_{i+1})$
is uniform on the square $(-\pi,\pi]^2$, of area $4\pi^2$. The two
conditions above cut off two corner triangles of this square, with legs of
lengths $\pi-c_i$ and $\pi+c_i$ respectively, hence of areas
$(\pi-c_i)^2/2$ and $(\pi+c_i)^2/2$ (both thresholds lie in $(0,2\pi)$
since $|c_i|<\pi$). Therefore
\[
p_i:=\PP(\text{bond $i$ crosses})
=\frac{(\pi-c_i)^2+(\pi+c_i)^2}{8\pi^2}
=\frac14+\frac{c_i^2}{4\pi^2}.
\]
Writing $N_n=\sum_i X_i$ with $X_i$ the indicator of a crossing at bond
$i$, we have $\E N_n=\sum_i p_i=n\,m_n$. Moreover $N_n$ is a function of
the independent coordinates $\delta_1,\dots,\delta_n$, and changing a
single $\delta_j$ alters at most the two indicators $X_{j-1},X_j$, hence
changes $N_n$ by at most $2$. McDiarmid's bounded-differences
inequality~\cite{McDiarmid,BoucheronLugosiMassart} then gives, for $t>0$,
\[
\PP\bigl(|N_n-\E N_n|\ge t\bigr)\le
2\exp\Bigl(-\frac{2t^2}{4n}\Bigr),
\]
and $t=\varepsilon n$ yields \eqref{eq:crossings}. For (i), $c_i\equiv 2\pi q/n$, so $\frac{1}{4\pi^2 n}\sum_i c_i^2
=q^2/n^2\to0$ when $|q|=o(n)$. For (ii), $N_n$ is a function of the $2n$ independent coordinates
$(\theta^*_1,\dots,\theta^*_n,\delta_1,\dots,\delta_n)$, each with
bounded-differences constant $2$, so the same argument gives
$\PP(|N_n-\E N_n|\ge\varepsilon n)\le 2e^{-\varepsilon^2 n/4}$, now with
$\E N_n/n=\E m_n=\frac14+\frac{\E c_1^2}{4\pi^2}
=\frac14+\frac1{12}=\frac13$ exactly, for every $n$.
\end{proof}

%% ============================================================
\subsection{Equidistribution along rays.}\label{sec:rays}
%% ============================================================

\begin{proof}[Proof of Theorem~\ref{thm:tentacles}]
The map $\lambda\mapsto\bm\theta^{(q)}+\lambda\bm v\pmod{2\pi}$ is a
linear flow on $\T^n$. By the Weyl--Kronecker equidistribution
theorem~\cite{Weyl1916,KuipersNiederreiter}, the orbit is
equidistributed on $\T^n$ with respect to $\mu$ if and only if
$\sum_ik_iv_i\neq0$ for every $\bm k\in\Z^n\setminus\{\bm 0\}$,
i.e.\ if $v_1,\dots,v_n$ are rationally independent. For each fixed
$\bm k\neq\bm 0$, the set $\{\bm v\in S^{n-1}:\sum_ik_iv_i=0\}$ is a
codimension-one subsphere, of spherical measure zero; the exceptional
set is the countable union over $\bm k$, still of measure zero.

Equidistribution of the orbit gives~\eqref{eq:equidist} for every
$\mu$-continuity set, that is, every Borel set whose topological
boundary is $\mu$-null. Each $\cK_{q'}$ qualifies: its boundary is
contained in the null set $\bigcup_i\{\eta_i=\pi\}$, a finite union
of hyperplane sections of $\T^n$. This proves~\eqref{eq:equidist}.
Finally, if $\mu(\cK_{q'})>0$ then by~\eqref{eq:equidist} the ray
spends a set of $\lambda$ of infinite Lebesgue measure inside
$\cK_{q'}$ and, by the same token applied to the complement, a set of
infinite measure outside it; hence it enters and exits $\cK_{q'}$
infinitely many times.
\end{proof}

%% ============================================================
\subsection{The head of the octopus.}\label{sec:head}
%% ============================================================

Throughout this section, $A$ denotes the cycle difference matrix
whose $i$th row is $\bm e_{i+1}-\bm e_i$ (indices mod $n$), so that
$\bm\eta=A\bm\theta$.

\begin{proof}[Proof of Theorem~\ref{thm:head}]
Write $\beta=2\pi q/n$, so the $q$-twisted state has all phase differences
equal to $\beta$, with $\beta\in(-\pi,\pi)$ since $|q|<n/2$. Along a ray
$\boldsymbol\theta^{(q)}+\lambda\bm v$ the phase differences are
$\eta_i(\lambda)=\beta+\lambda w_i$ with $\bm w=A\bm v$, and the ray
remains in $\mathcal K_q$ as long as $\eta_i(\lambda)\in(-\pi,\pi)$ for all $i$.
The first crossing is therefore
\begin{equation}
\lambda^*(\bm v)=\min_{i\,:\,w_i\neq0}
\frac{\pi-\beta\,\operatorname{sgn}(w_i)}{|w_i|},
\label{eq:lambda-star}
\end{equation}
the $i$th term being the value of $\lambda$ at which $\eta_i$ reaches the
nearer of the two walls $\pm\pi$ in the direction of travel.

{(i)} For the inscribed radius we minimize \eqref{eq:lambda-star} over unit
vectors. For each fixed $i$,
\[
\inf_{\|\bm v\|=1}\frac{\pi-\beta\,\operatorname{sgn}(w_i)}{|w_i|}
=\frac{\pi-|\beta|}{\displaystyle\sup_{\|\bm v\|=1}|w_i|},
\]
since one is free to choose the sign of $w_i$, and the smaller numerator
$\pi-|\beta|$ is obtained by aligning $\operatorname{sgn}(w_i)$ with
$\operatorname{sgn}(\beta)$. Now $w_i=(A\bm v)_i=\langle \bm e_{i+1}-\bm e_i,\bm v\rangle$,
so $\sup_{\|\bm v\|=1}|w_i|=\|\bm e_{i+1}-\bm e_i\|_2=\sqrt2$,
attained at $\bm v=\pm(\bm e_{i+1}-\bm e_i)/\sqrt2$ (which is
automatically orthogonal to $\mathbf 1$). Taking the infimum over $i$ as well,
\[
R_q=\inf_{\|\bm v\|=1}\lambda^*(\bm v)=\frac{\pi-|\beta|}{\sqrt2}
=\frac{\pi}{\sqrt2}\Bigl(1-\frac{2|q|}{n}\Bigr).
\]
Note that $R_q\to0$ as $|q|\to n/2$, the point at which $\boldsymbol\theta^{(q)}$
ceases to be stable (Corollary~\ref{cor:stability}) and that $R_q/\sqrt n\to0$, uniformly in $q$ as $n\to\infty$.

{(ii)} For a typical direction the dependence on $q$ is negligible. Since
$|q|=o(n)$ we have $\beta=2\pi q/n\to0$, so the two walls are $\pi\mp\beta=\pi(1+o(1))$ uniformly, and
$\lambda^*(\bm v)=\pi(1+o(1))/\|A\bm v\|_\infty$. It thus suffices to
analyze $\|A\bm v\|_\infty$ for $\bm v$ uniform on $S^{n-1}$. Use the Gaussian representation of the uniform direction:
let $\bm g=(g_1,\dots,g_n)$ be a standard Gaussian vector, so that
$\bm v=\bm g/\|\bm g\|_2$ is uniform on $S^{n-1}$ and
\begin{equation*}
\|A\bm v\|_\infty
= \frac{\max_{1\le i\le n}|\xi_i|}{\|\bm g\|_2},
\qquad \xi_i:=g_{i+1}-g_i .
\end{equation*}
The sequence $(\xi_i)_{1\le i\le n}$ (indices mod $n$) is a
stationary Gaussian sequence with $\E\xi_i=0$,
$\operatorname{Var}\xi_i=2$, and covariance
$\operatorname{Cov}(\xi_i,\xi_j)=-1$ for $|i-j|=1$ and $0$ for
$2\le|i-j|\le n-2$: it is $1$-dependent up to the cyclic wrap.
Classical extreme-value theory for stationary Gaussian sequences
under Berman's condition~\cite{LeadbetterLindgrenRootzen}---trivially
satisfied for $m$-dependent sequences---gives
\begin{equation*}
\frac{\max_{i}|\xi_i|}{\sqrt{2\cdot 2\log n}}
\ \to \ 1, \quad \text{in probability.}
\end{equation*}
while the law of large numbers gives
$\|\bm g\|_2/\sqrt n\to 1$, almost surely. Hence
\begin{equation*}
\|A\bm v\|_\infty\,\cdot\,\frac{\sqrt n}{2\sqrt{\log n}}
\ \to \ 1, \quad \text{in probability.}
\end{equation*}
and substituting into~\eqref{eq:lambda-star} yields
$\lambda^*(\bm v)\sqrt{\log n/n}\to\pi/2$ in probability.
\end{proof}

%% ============================================================
\section{Discussion and open problems}\label{sec:discussion}
%% ============================================================

The geometry established here does not depend on any specific choice of $f$ within hypotheses (H1)--(H3). It
is a consequence of two structural facts: a global integer invariant---the winding number---that the dynamics conserves, and the geometry of the uniform measure on $\T^n$ and the sets $\mathcal K_q$. The role of $f$ is
confined to the within-basin dynamics; basin volumes, the typical distance distribution, the head size, and the equidistribution of generic rays follow from the conservation law and the initial measure
alone. The results of \cite{zhang2021basins, groisman2025size} suggest that this universality extends at least to any $f$ which is increasing in an interval $(-a,a)$ centered at the origin and decreasing in the complement. In fact, Proposition \ref{prop:invariance} holds in this case for $\mathcal I=\{|\eta_i|< a, i=1,\dots, n\}$ independently of the shape of $f$. We also expect the universal properties to hold for a larger family of graphs that include at least grid discretizations of the torus in dimension $d\ge 2$ and random geometric graphs in the torus in any dimension. Numerical observations of similar geometry in
glasses~\cite{stillinger2015energy,charbonneau2017glass}, jammed
packings~\cite{ashwin2012calculations,martiniani2016structural,martiniani2017thesis,suryadevara2025basins, casiulis2023when}, and loss landscapes of neural
networks~\cite{choromanska2015loss,li2018visualizing} suggest that
the same template organizes octopus geometry far beyond oscillator
networks. Our results are the first rigorous instance of this universality picture.

The energy form $E(\bm\theta)=\sum_jF(\eta_j)$ is not isolated to
this model. The cluster dynamics of self-attention layers in
transformers are not a gradient flow but the associated ODE is structurally analogous
~\cite{geshkovski2024emergence}, and the AKOrN architecture
uses Kuramoto-type oscillator coupling explicitly as the forward pass
of a neural network~\cite{miyato2025akorn}. In both, what the network
selects at inference is fixed by basin geometry. The cycle is the
simplest topology in which a single integer invariant controls basin
labels; the same framework---topological invariants paired with the
geometry of the natural initial measure---extends naturally to
oscillator networks on richer
topologies~\cite{delabays2016multistability,jafarpour2022flow}.

To summerize, we close with the directions this work leaves open.
\begin{enumerate}
\item[(1)] \emph{Beyond the cycle.} Identify which features of graph
topology and coupling determine the universal behavior established
here, and quantify how far the rigorous picture transports to
landscapes carrying no integer invariant.
\item[(2)] \emph{Back to the sine coupling and beyond.} The geometric statements
of Theorems~\ref{thm:master}--\ref{thm:head} concern the basins at
$t=0$ and are not accessible to the entry-time analysis
of~\cite{groisman2025size}; deciding which of them survive for
$f=\sin$ and more general couplings, remains open.
\end{enumerate}

\section*{Acknowledgments}
The author thanks Yuanzhao Zhang, Cecilia De Vita, and Juli\'an
Fern\'andez Bonder for earlier collaborations that made this work
possible, Stefano Martiniani for pointing out the prior-art lineage
on octopus basins in jammed packings, Nicol\'as Garc\'ia Trillos for
the connection with transformer dynamics, and Steven Strogatz for
inspiring discussions. The author acknowledges the use of Claude
(Anthropic) for assistance with LaTeX drafting and editorial
polishing. The author assumes responsibility for all content. The Python code used to generate the figures is available upon reasonable request.

\bibliographystyle{amsplain}
\bibliography{bibli}

\end{document}